\documentclass[journal]{IEEEtran}
\usepackage[cmex10]{amsmath}
\usepackage{amssymb}
\usepackage{amsthm}
\usepackage{epsfig}
\usepackage{color}
\usepackage{algorithm}
\usepackage{algorithmic}
\usepackage{multirow}
\usepackage{array}
\newtheorem{theorem}{Theorem}

\usepackage{epsfig}
\usepackage{color}

\begin{document}

\title{Hidden Attacks on Power Grid: \\Optimal Attack Strategies and Mitigation}
\author{Deepjyoti Deka,\IEEEmembership{~Student Member,~IEEE,} Ross Baldick,\IEEEmembership{~Fellow,~IEEE,} and Sriram Vishwanath,\IEEEmembership{~Senior Member,~IEEE}
\thanks{Deepjyoti Deka, Ross Baldick and Sriram Vishwanath are with the Department of Electrical and Computer Engineering, The University of Texas at Austin, Austin, TX 78712 USA (e-mail:deepjyotideka@utexas.edu; baldick@ece.utexas.edu; sriram@ece.utexas.edu)}}

\maketitle
\begin{abstract}
 Real time operation of the power grid and synchronism of its different elements require accurate estimation of its state variables. Errors in state estimation will lead to sub-optimal Optimal Power Flow (OPF) solutions and subsequent increase in the price of electricity in the market or, potentially overload and create line outages. This paper studies hidden data attacks on power systems by an adversary trying to manipulate state estimators. The adversary gains control of a few meters, and is able to introduce spurious measurements in them. The paper presents a polynomial time algorithm using min-cut calculations to determine the minimum number of measurements an adversary needs to manipulate in order to perform a hidden attack. Greedy techniques are presented to aid the system operator in identifying critical measurements for protection to prevent such hidden data attacks. Secure PMU placement against data attacks is also discussed and an algorithm for placing PMUs for this purpose is developed. The performances of the proposed algorithms are shown through simulations on IEEE test cases.
\end{abstract}

\section{Introduction}
Stable and secure power grid operation relies on accurate monitoring of the state of its different components, including line currents and bus voltages. The state vector is estimated in a power grid by using a variety of measurement units in different buses and lines operating in the grid. These measurements are also used to facilitate operation of the electricity market through locational marginal pricing \cite{LMP}. Given their importance in the operation of the grid, error free delivery of data from the distributed meters to the central controller for state estimation is a critical need of every power grid. The effect of incorrect data collection had received attention in the $2003$ North-East blackout where incorrect telemetry due to an inoperative state estimator was listed as one of its principal causes \cite{2003}. Today Supervisory Control and Data Acquisition (SCADA) systems help relay the measurements to the state estimator of the grid for use in stability analysis and OPF solvers. However the presence of distributed meters spread across the entire geographical area covered by the power grid makes the grid vulnerable to cyber-attacks aimed at introducing malicious measurements. In fact, it has been reported that cyber-hacking had previously compromised the U.S. electric grid \cite{wallstreet} and can lead to sub-optimal electricity prices \cite{price}.

We consider here a scenario where an adversary in the power grid gains control of some meters in the grid and inject malicious data which can lead to incorrect state estimation. In reality, measurements collected have random noise due to measurement errors, but state estimators are able to overcome those through the use of statistical methods like maximum likelihood criterion and weighted least-square criterion \cite{estimation}. A coordinated attack on multiple meters by an adversary can evade detection by the standard mechanism present at the estimator and go unobservable and not raise any alarm at the state estimator.

Reference \cite{hidden} studies this problem of hidden attacks on the power grid which are not detected by tests on the residual of the measurements and shows that a few measurements are enough for the adversary to produce a hidden attack. The authors of \cite{hidden} also show that the set of protected measurements needed to prevent a hidden attack is of the same size as the number of state variables in the grid. Full protection from any hidden attack is thus very expensive. There have been multiple efforts aimed at studying the construction of malicious attack vectors for the adversary. In \cite{poor}, the authors discuss the optimal attack-vector construction for the constrained adversary using $l_0$ and $l_1$ recovery methods. A constrained adversary is governed by its objective to manipulate the measurements of the minimum number of meters to produce the desired errors in estimation. Reference \cite{sou} provides an approach for the creation of the optimal attack vector based on mixed integer linear programming. Such design approaches are NP-hard in general and hence require relaxations of the problem statement and provide approximate solutions at best. In addition, previous work such as that in \cite{thomas} requires certain assumptions on states of the system which may not hold in general.

In this paper, we consider an adversary with constrained resources. Following the attack model in \cite{poor}, we define the objective of the adversary as identifying the minimum number of meters that may be manipulated in order to create a hidden attack vector using those meters. We use graph-theoretic ideas like min-cut calculations in determining this optimal attack vector. Unlike previous work, we show that our solution does not require any assumption on the structure of the grid or any relaxation of the problem statement. The complexity of the algorithm for attack vector construction is shown to be polynomial in the number of nodes (buses) and edges (lines) in the power grid. Given the size of large power grids, polynomial running time of the algorithm justifies its significance when compared with NP-hard and brute force methods used in the existing literature.

In addition, the algorithm for attack vector construction does not depend on the exact values of the measurement matrix used in state estimation, relying primarily on the adjacency matrix of the network representing the power grid. This is significant, since the adjacency matrix of grids is often already known or can be approximated by an adversary from publicly available information. Using the novel graph theoretic framework discussed in this paper that requires only the adjacency matrix of the graph representing the power grid, the adversary can, thus, generate the optimal attack vector for malicious hidden attacks on the grid using a polynomial time algorithm.

We demonstrate that the power grid is significantly vulnerable to hidden attacks from even constrained adversaries with limited information (adjacency matrix in our case). We are unaware of any existing work providing optimal solution to this hidden attack problem in polynomial time. The framework developed is easily extended to formulate attack vectors in cases when certain measurements in the system are already protected or the state vector is partially known, i.e., some of the state variables are protected. We also provide algorithms to select, in a greedy fashion, critical measurements for protection to prevent a hidden attack given the prior knowledge of the adversary's resources (maximum number of measurements it can corrupt). Present power grids have additional meters called phasor measurement units (PMUs) installed on a few buses \cite{pmu1},\cite{pmu2}. Placement of a PMU at a given bus in the power grid provides measurement of the voltage phasor at that bus as well as the current flows of all lines incident on that bus. We discuss PMU placement in power grids in the context of hidden data attacks in detail in a separate section. The problems discussed in this context include designing an optimal attack vector for a grid equipped with PMUs, and the selection of locations for placement of PMUs in the grid in order to provide increased protection to the grid.

The main results of this paper are as follows:
\begin{itemize}
\item We provide a graph theoretic formulation for the problem of constructing an adversarial optimal attack vector with the minimum number of non-zero values that results in a hidden attack on the grid. Further, we present a algorithm that results in an exact solution to the problem and prove its optimality and polynomial-time complexity.

\item We discuss different variations of the adversarial attack problem, given the knowledge of certain prevailing protected measurements and state variables within the grid. We extend this discussion to power systems with PMUs installed at a few buses.

\item We present greedy algorithms to select additional protected measurements and locations for placement of PMUs in the grid, in order to hinder an adversarial hidden attack on the grid.

\item We study the performance of the algorithms through simulations on IEEE test bus systems and compare them with other algorithms in literature, as well as with brute force techniques.
\end{itemize}

The rest of this paper is organized as follows. The next section presents a description of the system model used in estimating the state variables in a power grid and formulation of the adversarial attack problem. The novel algorithm to determine the optimal solution to this problem, which includes selection of the measurements to attack and generating the attack vector, is discussed in Section \ref{sec:adversary solution}. Construction of the attack vector in the presence of protected measurements and state variables in the system is discussed in Section \ref{sec:variations}. Algorithms to select existing measurements in the system for protection to prevent hidden attacks on the grid are provided in Section \ref{sec:GM}. Power grids with installed PMUs are discussed in Section \ref{sec:PMU} wherein we design attack vectors for systems with PMUs and also analyze the placement of additional PMUs against hidden attacks for such systems. Simulations of the proposed algorithms on test IEEE bus systems and comparisons with other algorithms and brute force methods are reported in Section \ref{sec:results}. Finally, concluding remarks and future directions of work are presented in Section \ref{sec:conclusion}.

\section{Estimation in the Power Grid and Attack Models}
\label{sec:attack}
We represent the power grid using an undirected graph $(V,E)$, where $V$ represents the set of buses and $E$ represents the set of transmission lines connecting those buses. There are two kinds of measurements in the power grid in this model: flow measurements and voltage phasor measurements. Meter on a line in $E$ measure the power flow through that line while a meter on a bus in $V$ measure the voltage phasor at that bus. In addition, a PMU can collect both these kinds of measurements (bus voltage and current flows in all lines connected to that bus). This model is sufficiently general and can include cases where some bus voltages or line measurements are measured multiple times for redundancy. Figure \ref{14bus} shows the graph representation of the IEEE 14 bus test system. It can be found at \cite{testsystem}.

\begin{figure}
\centering
\includegraphics[width=0.4\textwidth]{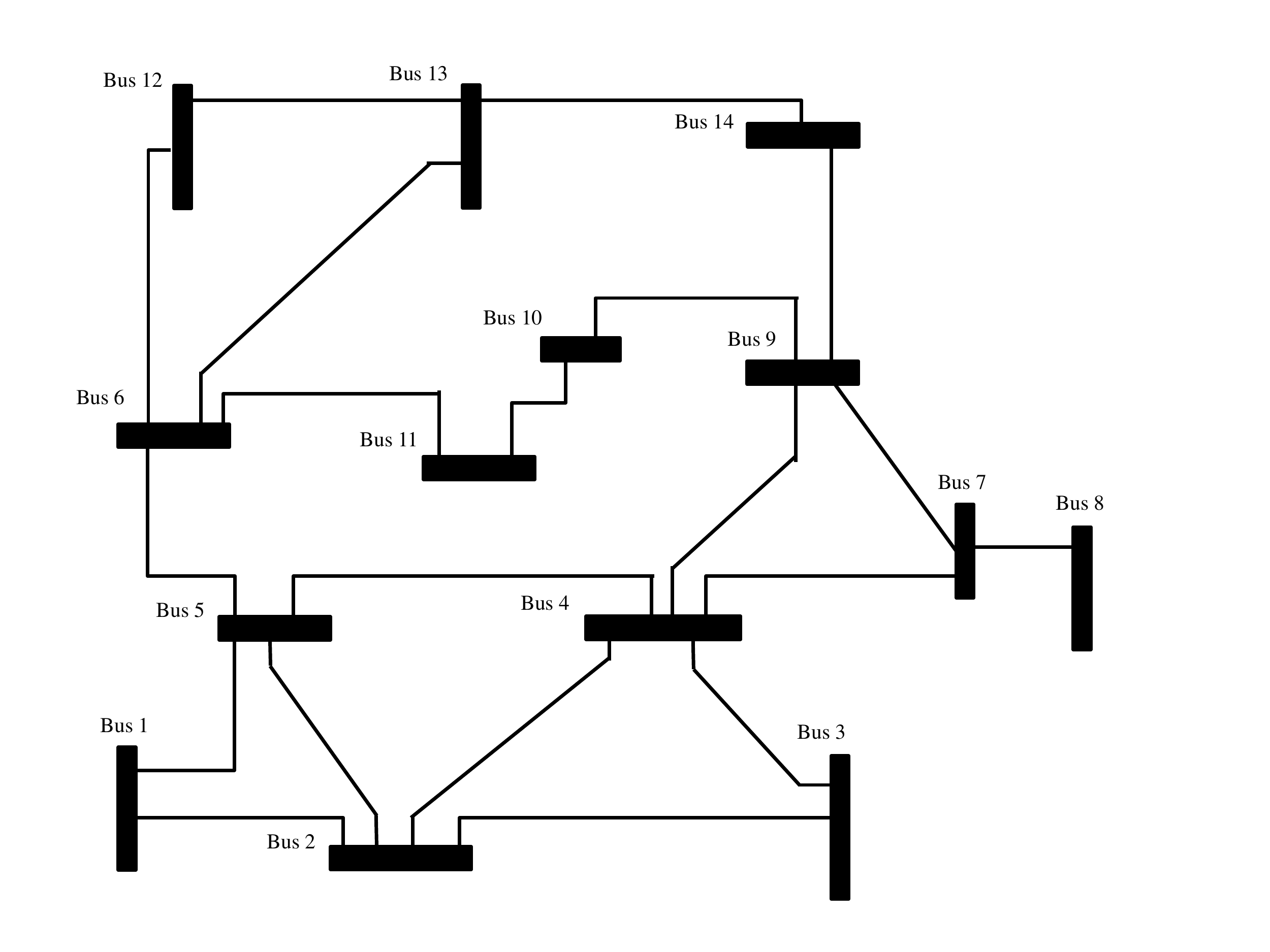}
\caption{IEEE 14-bus test system \cite{testsystem}}
\label{14bus}
\end{figure}

We consider the DC power flow model in a power grid
\begin{align}
z = Hx + e
\end{align}
where $z \in \mathbb{R}^m$ is the vector of measurements. $x \in \mathbb{R}^n$ is the state vector and consists of the voltage phase angles at the buses in the grid. The voltage magnitudes are taken as unity in the DC model. $H$ is the measurement matrix which relates the measurements with the state vector and $e$ is the zero mean Gaussian noise vector associated with the measurements. In general, $m>n$ which implies that there are more measurements than state variables to help provide redundant measurements to the state estimator. Here, $H$ depends on the topology of the network, the location of the meters as well as on the parameters of transmission lines (like resistance and susceptance). The DC power flow on a line $(i,j)\in E$ between buses $i$ and $j$ is given by $B_{ij}(x_i-x_j)$, where $B_{ij}$ is the magnitude of susceptance of the line $(i,j)$. If the $k^{th}$ measurement corresponds to this flow in line $(i,j)$, then $z_{k}$ is given by:
\begin{align}
z_{k} = H_{k}x = B_{ij}x_i-B_{ij}x_j. \label{flow}
\end{align}
The associated row in the measurement matrix ($H_k$) is a sparse vector with two non-zero values, $B_{ij}$ at the $i^{th}$ position and $-B_{ij}$ at the $j^{th}$ position.
\begin{align}
H_{k} = [0..0~~B_{ij}~~ 0..0~~-B_{ij}~~0..0] \label{line}
\end{align}
If $z_m$ measures the voltage phase angle at bus $i$ (the voltage magnitude is considered unity in the DC power flow model), the corresponding row $H_m$ in the measurement matrix is given by:
\begin{align}
&z_{m} = H_{m}x = x_i\\
\Rightarrow~&H_{m} = [0..0 ~~ 1~~ 0..0] \label{bus}
\end{align}
Here $H_m$ is a sparse vector with one in the $i^{th}$ position and zero everywhere else. The measurement matrix is, thus, very sparse with a maximum of $2$ non-zero values per row. In order to enable correct state estimation, the measurement matrix must have full column rank of $n$. The state estimator uses the measurements and outputs the estimated state vector $\hat{x}$ by minimizing the residual $\|z-H\hat{x}\|_2$. In normal operation, the magnitude of the minimum residual is smaller than an established threshold which monitors the correctness of the estimate. \\Given such an estimator, the adversary corrupts the measurement vector $z$ by adding an attack vector $a$ to generate the new measurements of the form $\hat{z} = z+a$. It is fairly straightforward (\cite{hidden}) that, if $a$ satisfies
\begin{align}
a = Hc
\end{align}
for some $c \in \mathbb{R}^n$, then the residual calculated by the estimator remains the same as before. Therefore the estimator remains incapable of detecting the presence of an attack, outputting an erroneous state vector estimate $\hat{x} +c$.

Attack vector construction refers to the creation of an optimal attack vector $a$ to corrupt the measurements. Reference \cite{hidden} creates such attack vectors by using a projection matrix $P$ using the measurement matrix $H$. A major problem of interest here is the case of a constrained adversary with limited resources. Such an adversary attacks the minimum number of measurements to create a successful hidden attack. Here, an attack is considered successful if at least $1$ state variable suffers a change in magnitude after the attack. Equivalently, the construction of the attack vector $a^*$ can be characterized as a solution of the following optimization problem :
\begin{align} \label{opt_attack}
&\min_{a} \|a\|_{0} \nonumber \\
\text{s.t. ~} & a =Hc, ~c \neq \vec{0}
\end{align}
This is similar to the attacker's problem in \cite{poor}, where the constraint $c \neq \vec{0}$ is replaced by the constraint $\|c\|_\infty \geq \tau$. Both these problem formulations are essentially the same as every solution of Problem (\ref{opt_attack}) can be suitably scaled to result in a solution obtained in \cite{poor}. For every non-zero $c$, $c\tau/\|c\|_\infty$ satisfies the constraint $\|c\|_\infty \geq \tau$.

In the next section, we present our algorithm for designing an optimal attack vector using graph theory. Unlike \cite{poor}, we do not need the exact matrix $H$ for our solution, but only the locations of $1$s in it.

\section{Optimal Attack Vector Design}
\label{sec:adversary solution}
Consider the $m$ times $n$ measurement matrix $H$ as noted in (\ref{line}) and (\ref{bus}). It is sparse and every row has either $1$ (corresponding to phase angle measurement or $2$ (corresponding to line flow measurement) non-zero elements. We first augment one extra column $h^{g}$ to the right of the matrix $H$ to create a $m$ times $(n+1)$ modified measurement matrix $\hat{H}$ such that for every row $H_m$ with a phase angle measurement, the new column $h^g$ has a value of $-1$ at the $m^{th}$ location.
\begin{align}
\hat{H}_{m} = [H_m ~|~ -1] \label{bus1}
\end{align}
For a line flow measurement $H_k$, the corresponding element in the new column is $0$.
\begin{align}
\hat{H}_{k} = [H_k ~|~ 0] \label{line1}
\end{align}
The state vector $c$ is now augmented to form a new state vector $\hat{c} = \setlength{\arraycolsep}{2pt} \renewcommand{\arraystretch}{0.8}\begin{bmatrix} c \\0 \end{bmatrix}$ which has $0$ as the final element. We now have
\begin{align}
&a= Hc = \hat{H}\hat{c} = \left[H ~|~ h^g\right]\begin{bmatrix} c \\0 \end{bmatrix}
\label{modify}
\end{align}
Equation (\ref{modify}) holds as the last element of $\hat{c}$ is $0$. The new formulation ($\hat{H},~ \hat{c}$) provides a reference phase angle of $0$ (represented by the extra element in $\hat{c}$) such that phase angle measurement at a bus becomes equivalent to a line flow measurement between the bus and reference phase angle. Note that every row in $\hat{H}$ has $2$ non-zero elements and corresponds to a line flow measurement now. Next, we state and prove a theorem which will enable us to develop the algorithm for attack vector.

\begin{theorem} \label{01}
There exists a non-zero binary $0-1$ vector $c_{opt}$ of size $n$ times $1$ for the optimal attack vector $a^*$ given by Problem (\ref{opt_attack}) such that $\|a^*\|_0 = \|Hc_{opt}\|_0$.
\end{theorem}

\begin{proof}
Consider Problem (\ref{opt_attack}). Let the optimal attack vector be given by $a^* = Hc^*$. If $c^*$ is a $0-1$ vector, take $c_{opt} =c^*$ and the theorem is trivially true. If $c$ is not a $0-1$ vector, construct $n$ times $1$ vector $c_{opt}$ such that $c_{opt}(i) = \textbf{1}(c^*(i) \neq 0)$, $\forall i \in \{1, n\}$. Consider $\|Hc_{opt}\|_0 $. For every non-zero phase angle measurement in $Hc^*$, we have a corresponding non-zero measurement in $Hc_{opt}$. However, in case of a non-zero value of line flow measurement in $Hc^*$ between two neighboring buses $i$ and $j$ with $c^*(i) \neq 0, c^*(j) \neq 0$, the corresponding value of $Hc_{opt}$ is $0$ as $c_{opt}(i) = c_{opt}(j) = 1$. It, thus, follows from the structure of the $H$ matrix that $\|Hc_{opt}\|_0 \leq \|Hc^*\|_0 $. Since $a^* = Hc^*$ is the optimal attack vector with minimum number of non-zero entries, we have $\|Hc_{opt}\|_0 =\|Hc^*\|_0$. Thus, $Hc_{opt}$ also gives an optimal attack vector and $\|a^*\|_0 = \|Hc_{opt}\|_0$.
\end{proof}

Now, consider the modified measurement matrix $\hat{H}$ and the augmented vector $\hat{c}$. Using Equation (\ref{modify}) and the previous theorem, we conclude that the optimal attack vector is given by $a^* = \hat{H}\hat{c}_{opt}$ where $\hat{c}_{opt}$ is the $0-1$ vector given by $\hat{c}_{opt} = \setlength{\arraycolsep}{2pt} \renewcommand{\arraystretch}{0.8}\begin{bmatrix} c_{opt} \\0 \end{bmatrix}$.

Minimizing the number of measurements needed by the adversary to inject the optimal attack vector is equivalent to minimizing the number of non-zero flow measurements given by $\hat{H}\hat{c}$. Since we are concerned only with $\|a\|_0$ and $\hat{c}$ is a $0-1$ vector, we observe that the exact values of susceptance present in $\hat{H}$ are not needed to get the optimal attack vector. In fact, the contribution of a flow measurement of $\hat{H}\hat{c}$ in $\|a\|_0$ will remain the same even if the susceptance $B_{ij}$ of every line included in $\hat{H}$ is changed to $1$. We, therefore, create an incidence matrix $A_H$ of dimension $m$ x $(n+1)$ from $\hat{H}$ by replacing every positive element in $\hat{H}$ with a $1$ in $A_H$ and each negative element in $\hat{H}$ with to a $-1$ in $A_H$. Zeroes are left unchanged. The $(i,j)^{th}$ elements of $\hat{H}$ and $A_H$ are related as
\begin{align}
A_H(i,j) = 1(\hat{H}(i,j) > 0) - 1(\hat{H}(i,j) < 0) \label{incident}
\end{align}
The main result of this paper which gives the minimum attack vector for the optimization Problem (\ref{opt_attack}) is given in the following theorem involving $A_H$.

\begin{theorem} \label{main}
The cardinality of the optimal attack vector in Problem (\ref{opt_attack}) with measurement matrix $H$ is equal to the min-cut of the undirected graph of $n+1$ nodes and edges defined by the incidence matrix $A_H$.
\end{theorem}
\begin{proof}From the discussion above, it is clear that for a given $0-1$ vector $\hat{c}$, $\|\hat{H}\hat{c}\|_0 = \|A_H\hat{c}\|_0$. Further, using Theorem \ref{01}, the optimization Problem (\ref{opt_attack}) can be written as
\begin{align} \label{opt_attack1}
&\min_{a}  \|a\|_{0}  \\
\text{s.t. ~} & a =A_H\hat{c}, ~\hat{c} \neq \vec{0} \nonumber\\
&\text{$\hat{c}$ is a $0-1$ vector with $(n+1)^{th}$ element $0$}\nonumber
\end{align}
This is the classical min-cut partition problem in graph theory. The minimum value of $\|q\|_0$ is, thus, given by the magnitude of the min-cut of the undirected graph with $A_H$ as the incidence matrix.
\end{proof}
Note that multiple measurements of the same line-flow or phase angle will lead to multiple edges between two nodes in the associated graph. Formally, after pre-processing the initial measurement matrix $H$ to generate $\hat{H}$ and $A_H$, the optimal attack vector and its cardinality is given by Algorithm $1$.

\begin{algorithm}
\caption{Optimal Attack Vector ($a^*$) through Min-Cut}
\textbf{Input:} Graph $G_H$ with incidence matrix $A_H$ \\
\begin{algorithmic}[1]
\STATE Compute the min-cut of the graph $G_H$
\STATE $c \gets \mathbf{1}$
\STATE Choose $(n+1)^{th}$ node as root
\STATE Remove min-cut edges
\STATE Do breadth first path traversal from root
\IF{ node $i$ is reached}
\STATE $c(i) \gets 0$
\ENDIF
\STATE $a^* \gets Hc $
\end{algorithmic}
\end{algorithm}

The optimal attack vector consists of the edges in the min-cut and produces a non-zero change in the estimate of the state variables at the nodes which are on the opposite side of the min-cut as the reference node. The resulting attack vector is indeed optimal as its cardinality is equal to the min-cut. The min-cut computation is a well-studied problem in graph theory and has a running time polynomial in the number of nodes and edges in the graph \cite{mincut1}. Reference \cite{mincut2} gives a simple algorithm for computing the min-cut in $O(|V|log|V| + |E|)$ time-steps. Here, $|V|$ and $|E|$ represent the number of nodes and edges in the graph considered.
The algorithm presented above, to find the optimal attack vector, has the following distinguishing characteristics which separate it from other algorithms in literature:
\begin{itemize}
\item It finds the optimal solution of the optimization Problem (\ref{opt_attack}) without using a relaxation
\item It is polynomial-time solvable
\item It does not require the exact values of the line susceptance in the grid, using instead, the locations of the measurements in the network.
\end{itemize}

Next, we show how the Algorithm $1$ can be used to design the optimal attack vector in the presence of protected measurements and state variables in the system.

\section{Optimal Attack Vector Construction with Protected Measurements and State Variables}
\label{sec:variations}
 Certain measurements in the power grid are protected from cyber-attacks by encryptions or by geographical isolation. This imposes a constraint on the adversary by requiring that the values of the attack vector $a$ corresponding to protected measurements be made $0$. Similarly, certain state variables might be protected from adversarial contamination due to the presence of secure channels of collecting their values. Let $S_m$ be the set of protected measurements and $S_v$ be the set of protected state variables. The optimal attack vector $a^*$ here can be written as the solution of the following optimization problem:
\begin{align} \label{opt_attack2}
&\min_{a} \|a\|_{0} \\
\text{s.t. ~} & a =Hc,~ c \neq \vec{0} \nonumber\\
& H^{S_m}c = 0, ~c(i) = 0 ~\forall i \in S_v\nonumber
\end{align}
where $H^{S_m}$ represents the rows in the measurement matrix corresponding to the protected measurements. Here, the adversary needs to ensure that the vector $c$ has values $0$ for the protected state variables, while the vector $a$ has values $0$ for the protected measurements.

Following Problem (\ref{opt_attack}), we consider the modified measurement matrix $\hat{H}$ (given by Equations (\ref{line}) and (\ref{bus})) and the augmented state vector $\hat{c}=\setlength{\arraycolsep}{2pt} \renewcommand{\arraystretch}{0.8}\begin{bmatrix} c \\0 \end{bmatrix}$ by adding the reference node. We first obtain the incident matrix $A_H$ from the modified measurement matrix $\hat{H}$ as per Equation (\ref{incident}). We denote the graph represented by the incident matrix $A_H$ as $G_H$. Every measurement in $A_H$ leads to an edge in $G_H$ of unit weight. The additional constraints due to the protected measurements and state variables are incluuded in the graph $G_H$ through the following modification as follows:

\noindent 1. Create an edge of infinite weight between the buses with protected state variables in $S_v$ and the reference node.

\noindent 2. Change the weights of edges with protected measurements in $S_m$ to infinity.

The resultant graph generated from $G_H$ after this modification is denoted by $G_H^*$.
We, now, run the steps outlined in Algorithm $1$ on $G_H^*$ to obtain the optimal attack vector as described in the previous section. We call this Algorithm $2$ for completion.
\begin{algorithm}
\caption{Optimal Attack Vector ($a^*$) with protected measurements and state variables}
\textbf{Input:} Graph $G_H$, protected state variables $S_v$ and measurements $S_m$  \\
\begin{algorithmic}[1]
\STATE Modify $G_H$ to generate $G_H^*$
\STATE Run Algorithm $1$ on $G_H^*$
\end{algorithmic}
\end{algorithm}

In the solution of Algorithm $2$, the modified edges (with infinite weight) are not included in the attack vector given by the min-cut to keep the value of the min-cut below infinity. This ensures that the modification of $G_H$ to $G_H^*$ satisfies the constraints arising due to protection and gives the optimal solution.\\
\noindent\textbf{$l_1$ Relaxation:} Problem (\ref{opt_attack2}) can also be relaxed and approximately solved using a na\"{i}ve $l_1$ relaxation by replacing the non-convex $l_0$ terms with $l_1$ terms \cite{sparsity}. However, such an approach leads to attack vector solutions with large cardinality which are sub-optimal. To go around that, we use thresholds in the formulation shown below to solve Problem \ref{opt_attack2}:
\begin{align} \label{opt_attack2modify}
&\min_{a} \|a\|_{1}  \\
\text{s.t. ~} & a =Hc,~ c \geq \vec{0},~ 1^T c > \theta_1\nonumber \\
& H^{S_m}c = 0, ~ c(i) = 0 ~\forall i \in S_v\nonumber
\end{align}
The final attack vector is obtained by thresholding the optimal solution $a^*(i) = 1(a^*(i)> \theta_2), \forall ~1 \leq i \leq m$, where $m$ is the length of the attack vector. Here, $\theta_1$ and $\theta_2$ are thresholds used to decrease the cardinality of the solution attack vector. $\theta_1$ and $\theta_2$ are taken as $1$ and $10^{-3}$ respectively in our simulations in Section \ref{sec:results}.

Until now, we have discussed the adversary's strategy for designing attack vectors towards causing hidden data attacks on the grid. We will use this knowledge in the next section to discuss policies that can be adopted by the power grid system operator or controller to restrict the efficacy of hidden attacks.

\section{Protection Strategies against Hidden Attacks}
\label{sec:GM}
Let us consider the system described in Problem (\ref{opt_attack2}). Here, there are pre-existing secure measurements (set $S_m$) and state variables (set $S_v$) in the grid. The corresponding set of unprotected measurements and unprotected state variables will be termed $S_m^c$ and $S_v^c$ respectively. For complete protection against hidden attacks, it has been shown in \cite{hidden} that $H^{S_m}$ should have full column rank. In that case, the only $c$ satisfying the constraint $H^{S_m}c = 0$ is the all-zero vector. However, for full column rank of $n$, the number of protected measurements needs to be greater than $n$ \cite{hidden}, and incurs a great cost. Instead, we look at the problem of augmenting the set of protected measurements $S_m$ with $k$ measurements selected from the unprotected set $S_m^c$. Protecting additional measurements leads to an increase in the cardinality of the optimal attack vector $\|a^*\|_0$. This increases the number of compromised measurements needed by the adversary for a successful attack. We formulate this problem as follows:
\begin{align} \label{opt_attack3}
\max_{S^* \in S_m^c}&\min_{a} \|a\|_{0}  \\
\text{s.t. ~} & a =Hc,~ c \neq \vec{0} \nonumber \\
& H^{S_m}c =0,~ H^{S^*}c = 0 \nonumber\\
& c(i) = 0 ~\forall ~i \in S_v,~ |S^*| = k \nonumber
\end{align}
The set of new protections $S^*$ of cardinality $k$ is then used to update the protected set $S_m$. As mentioned earlier, $S_m$, $S_v$, $S_m^c$ and $S_v^c$ represent the sets of protected measurements, protected state variables, unprotected measurements and unprotected state variables in the grid respectively.

Protecting optimal $k$ additional measurements is equivalent to increasing the weights of $k$ edges in the modified graph $G_H^*$ (outlined in Section \ref{sec:variations}) to infinity to maximally increase the value of the min-cut. This is a NP-hard problem. A brute- force selection of measurements for protection is computationally intensive and impractical given the large number of candidate measurements in the set $S_m^c$ in a real power grid. Hence, we provide here a greedy approach for Problem \ref{opt_attack3} in Algorithm $3$. Here, $S_m$ is updated in $k$ steps. At each step, the best candidate is chosen in a greedy fashion for protection given $a^*$, the current optimal attack vector. After including a measurement in the protected set $S_m$, $a^*$ is updated and used for selecting the next candidate measurement for protection.

\begin{algorithm}
\caption{Greedy Solution for Additional Protection}
\textbf{Input:} Graph $G_H^*$, attack vector $a^*$, protected set $S_m$ \\
\textbf{Output:} Updated $G_H^*$, $a^*$ and $S_m$\\
\begin{algorithmic}[1]
\FOR {$i = 1$ \TO $k$}
\STATE $a_{cm} \gets a^*$
\FOR {$j=1$ \TO $m$ \COMMENT{$m$: total measurements }}
\IF {$a^*(j) \neq 0$} \label{red}
\STATE $G_{temp}\gets G_H^*$
\STATE Protect measurement $j$ in $G_{temp}$
\STATE Compute optimal attack vector $a_{temp}$ for $G_{temp}$
\IF {$\|a_{temp}\|_0 \geq \|a_{cm}\|_0$}
\STATE $cm \gets j$    \COMMENT{current best candidate}
\STATE $a_{cm} \gets a_{temp}$ \COMMENT{current optimal attack vector}
\ENDIF
\ENDIF
\ENDFOR
\STATE Protect measurement $cm$ and update $G_H^*$
\STATE $a^* \gets a_{cm}$,~  $S_m \gets S_m \cup \{cm\}$
\ENDFOR
\end{algorithmic}
\end{algorithm}

Step \ref{red} of Algorithm $3$ retains only the measurements represented by the current min-cut of $G_H^*$ as candidates for the next update in $S_m$. It ignores measurements outside the current min-cut as protecting them does not lead to an increase in the size of the min-cut of the updated graph. This step, thus, leads to an reduction in the number of possible candidates in each step from $m - |S_m|$ to $\|a\|_0$ without any loss of performance. The Algorithm is of course sub-optimal compared to a computationally intensive brute force search of the best measurements for protection.

\section{Power Systems with PMUs: Attacks and Protection}
\label{sec:PMU}
 In this Section, we extend the ideas developed in the previous sections to power grids with Phasor Measurement Units (PMUs). A PMU located at a bus in the grid measures its voltage phasor as well as the current flows of all lines incident on that bus \cite{PMUwork}. Previous work on PMU placement against hidden measurement attacks \cite{poor}, \cite{thomas} assume full protection of the PMU's measurements. Recently, it has been shown that PMUs do not have full security as they rely on civilian GPS signals for real-time signalling that can thus be corrupted by GPS spoofers \cite{todd}. Therefore, we consider both protected and unprotected PMU measurements here.
\subsection {Optimal Attack Vector for Grids with PMUs}
The bus phase angles and line flow measurements calculated by the unsecured PMUs are equivalent to other measurements in the grid. We assume here that any measurement in an unsecured PMU can be independently corrupted by an adversary. For secure PMUs, we consider all measurements recorded by them as protected and include them in the protected set $S_m$ and the protected bus phase angles in the set of protected state variables $S_v$. The attack vector is then given by running Algorithm $2$. The optimality of the attack vector follows from the discussion for a general grid with protected measurements given in Sections \ref{sec:adversary solution} and \ref{sec:variations}.
\subsection {Protection against hidden attack by placing secure PMUs}
In this case, we consider secure PMUs such that their measurements are protected against any malicious attack. Each PMU placed at a bus thus creates protected measurements of the bus phase angle and incident line flows on that bus. Optimal Placement of secure PMUs to ensure full protection of all state variables against any adversary is equivalent to a set-cover problem and is NP-hard in general. However, approximate and distributed algorithms based on belief propagation have been shown to provide optimal PMU placement for several IEEE test systems \cite{deka2011}. Here, instead of full protection, we look at the problem of placing additional $k$ secure PMUs to maximally hinder a hidden attack by the adversary. We consider existing protected measurements and protected state variables in the grid and denote them by $S_m$ and $S_v$ respectively. It is worth noting that $k$ PMUs might not be sufficient to provide full protection to the entire state vector. Thus, we look at maximizing the cardinality of the optimal attack vector $a^*$ of the adversary instead. As discussed in Section \ref{sec:GM}, this is done by maximizing the min-cut of the modified graph $G_H^*$ associated with the measurement matrix $H$ and protected sets $S_m$ and $S_v$. We modify Algorithm $3$ and provide a greedy algorithm, Algorithm $4$, to determine $k$ bus locations, one at a time for placing secure PMUs. This greedy algorithm runs $k$ times and thus has a small complexity compared to a brute force search. The performance of the algorithm on IEEE test cases is reported in the following section.

\begin{algorithm}
\caption{Greedy Solution for $k$ secure PMU placement}
\textbf{Input:} Graph $G_H^*$, attack vector $a^*$, protected set $S_m$ \\
\textbf{Output:} Updated $G_H^*$, $a^*$\\
\begin{algorithmic}[1]
\FOR {$i = 1$ \TO $k$}
\STATE $a_{cm} \gets a^*$
\FOR {$j=1$ \TO $n$ \COMMENT{$m$: total buses}}
\STATE $G_{temp}\gets G_H^*$
\STATE Place PMU at bus $j$ in $G_{temp}$
\STATE Compute optimal attack vector $a_{temp}$ for $G_{temp}$
\IF {$\|a_{temp}\|_0 \geq \|a_{cm}\|_0$}
\STATE $cm \gets j$    \COMMENT{current best candidate bus}
\STATE $a_{cm} \gets a_{temp}$ \COMMENT{current optimal attack vector}
\ENDIF
\ENDFOR
\STATE Place PMU at bus $cm$ and update $G_H^*$
\STATE $a^* \gets a_{cm}$
\ENDFOR
\end{algorithmic}
\end{algorithm}

\section{Simulations on IEEE test systems}
\label{sec:results}
In this section, we evaluate the performance of our proposed algorithms by simulating their performance on different IEEE test bus systems, namely 14-bus, 30-bus and 118-bus systems. Data about these test systems can be found at \cite{testsystem}. All simulations are run in Matlab Version 2009a. We start by discussing the performance of Algorithms $1$ and $2$ that are used for constructing the optimal attack vector ($a^*$) of the adversary. We take IEEE-14 bus system and place flow measurements in all lines and voltage measurements on random $60\%$ of the buses. Figure \ref{attack14} shows the increase in the average size of the optimal attack vector with increase in the fraction of randomly protected measurements placed in the system. We see that plots generated by simulations of our algorithms and that obtained through brute force search for the optimal attack vector overlap completely, depicting optimal performance. It can be seen from the same figure that the performance of our algorithms are much better than the output of the $l_1$ relaxation given in Problem (\ref{opt_attack2modify}). Next, we plot the output of Algorithm $2$ and show its improved performance over the output of a $l_1$ relaxation approach for 30, 57 and 118 IEEE test bus systems under different system conditions in Figure \ref{attackall}. The output of $l_1$ relaxation is not close to optimal as commendable performance of $l_0-l_1$ solvers requires the measurement matrix to satisfy certain necessary conditions \cite{sparsity}. Such conditions are difficult to satisfy for test-systems where the network and the measurement matrices are not random, as in our case.

We now present results on our approach to Problem (\ref{opt_attack3}), which is given by Algorithm $3$. Here, we select, in a greedy fashion, the $k$ best measurements for protection such that the minimum number of measurements needed for a successful hidden attack increases the most. Figure \ref{measprot14} shows the performance of our greedy Algorithm $3$ for different values of $k$ for the IEEE-14 bus system with flow measurements on all lines, voltage measurements on $60\%$ of the buses and $1/6$ of measurements initially protected. We observe that the performance of the greedy algorithm in increasing the size of the optimal attack vector is comparable to a computationally intensive brute-force selection for protecting additional measurements in this case. We also simulate Algorithm $3$ for IEEE 30, 57 and 118 bus systems and plot the average improvement in the cardinality of the optimal attack vector with an increase in the value of $k$ in Figure \ref{measprotgreedy}. It is important to note that the minimum cardinality of the optimal attack vector $a^*$ does not increase significantly with a small increase in $k$ or fraction of protected measurements in all the different test systems considered. This observation can be explained using the fact that the test-systems are sparse and have several buses with low degree and thus have a low min-cut.
Determination of optimal attack vectors in the presence of secure PMUs for the IEEE 30 and 57 bus systems is shown in Figure \ref{attackpmu}. In either test system, we place line flow measurements in each bus and phase angle measurements in $60\%$ of the buses. We observe an expected increase in the average size of the optimal attack vector on increasing the fraction of buses randomly selected for placement of secure PMUs. Finally, we show the performance of Algorithm $4$ in the IEEE 30-bus system. In the base case, we put line flow measurements in each bus of the system, phase angle measurement in random $60\%$ of the buses and protect $1/10$ of the measurements selected randomly. Algorithm $4$ is used to place $k$ additional PMUs on buses to increase the cardinality of the optimal attack vector for the system. We observe again that enough secure PMUs need to be placed in the grid to significantly increase the cardinality of the optimal attack vector as high sparsity of the network graph and low degrees of the nodes keep the graph min-cut low.

\begin{figure}[ht]
\centering
\includegraphics[width=0.4\textwidth]{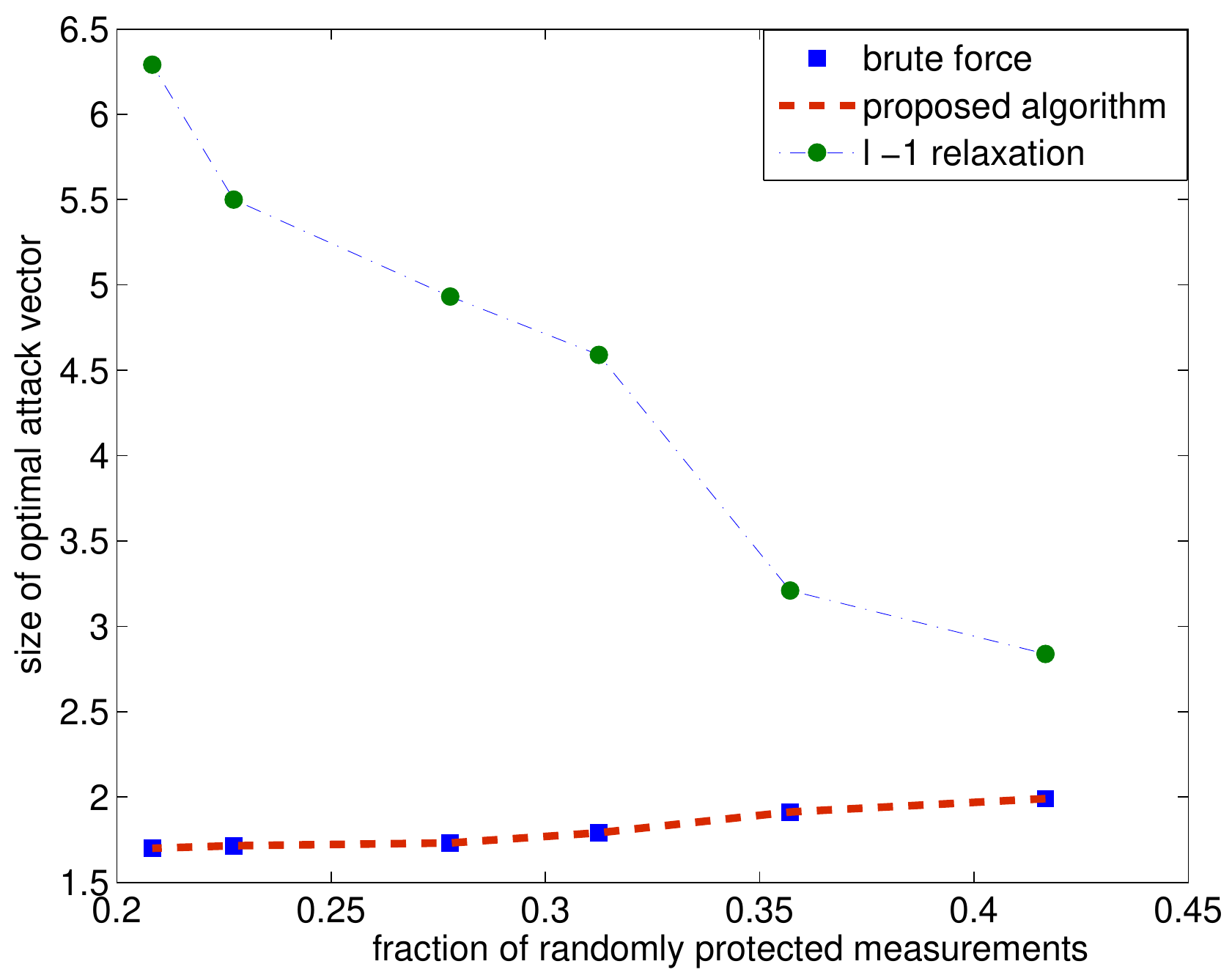}
\caption{Optimal hidden attack on IEEE 14-bus system with flow measurements on all lines, voltage measurements on $60\%$ of the buses and fraction of measurements randomly protected}
\label{attack14}
\end{figure}
\begin{figure}
\centering
\includegraphics[width=0.4\textwidth]{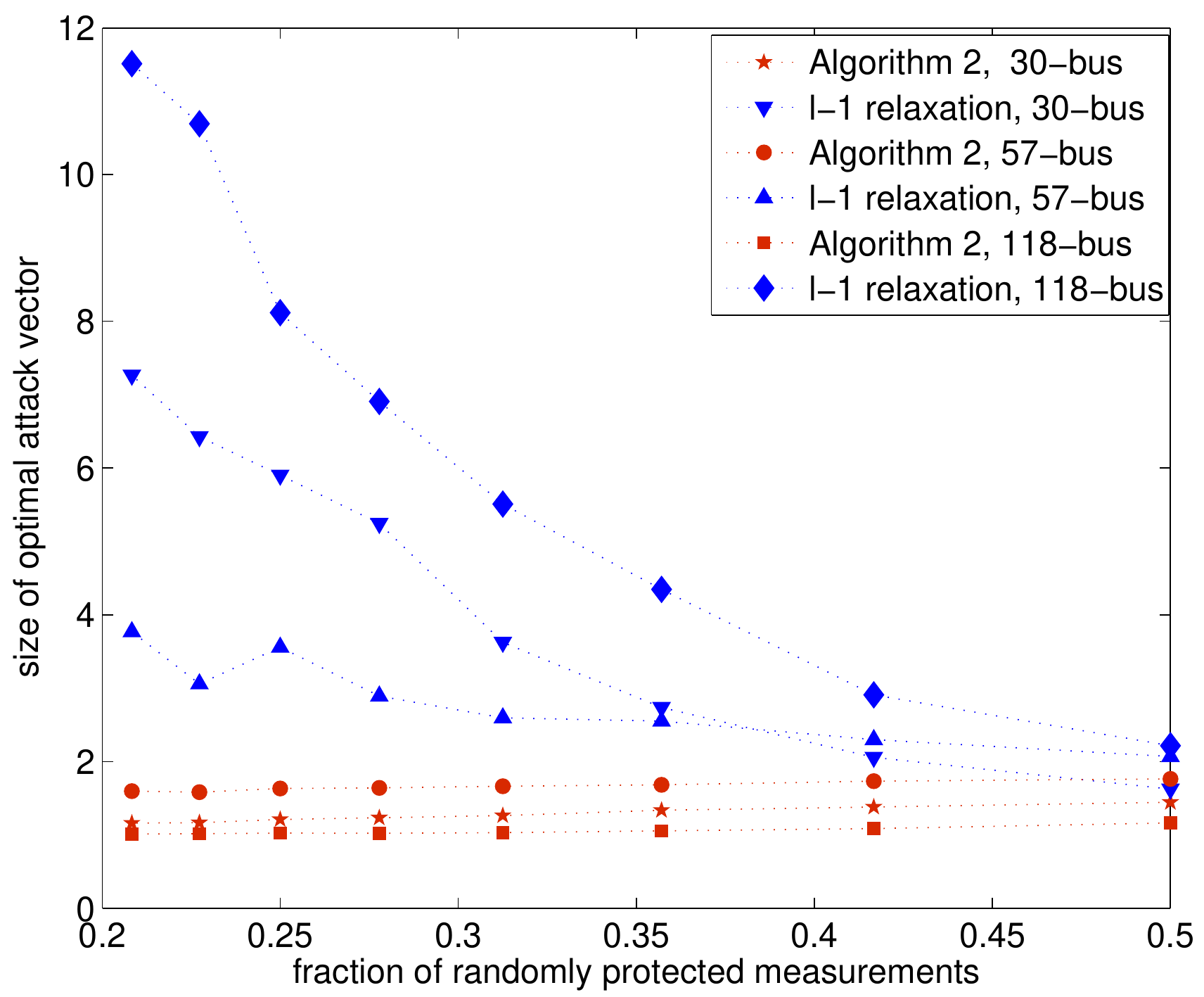}
\caption{Optimal hidden attack on IEEE test systems with flow measurements on all lines, voltage measurements on $60\%$ of the buses and fraction of measurements randomly protected}
\label{attackall}
\end{figure}
\begin{figure}
\centering
\includegraphics[width=0.4\textwidth]{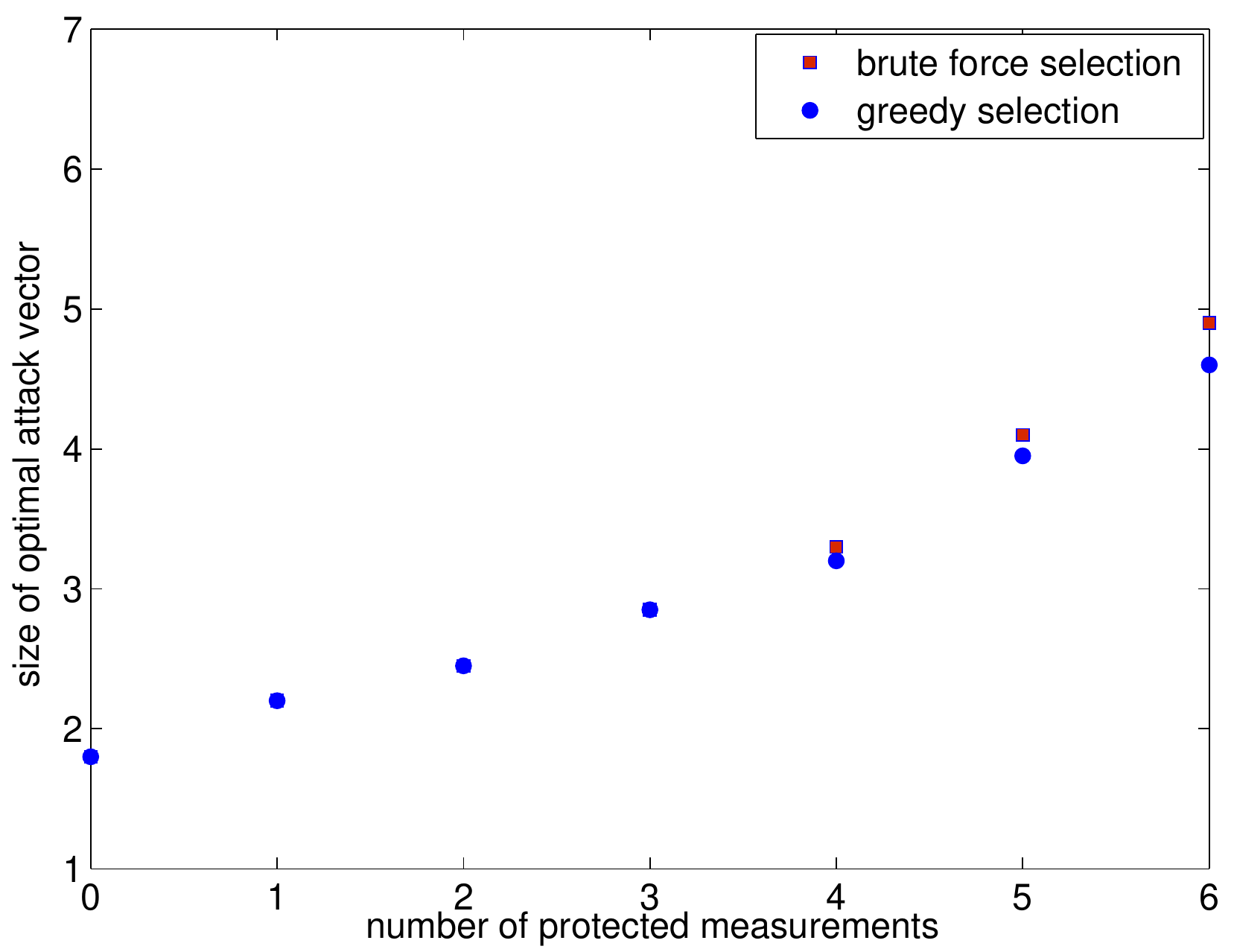}
\caption{Protection of additional measurements in IEEE 14-bus system with flow measurements on all lines, voltage measurements on $60\%$ of the buses and $1/6\%$ of measurements initially protected}
\label{measprot14}
\end{figure}
\begin{figure}
\centering
\includegraphics[width=0.4\textwidth]{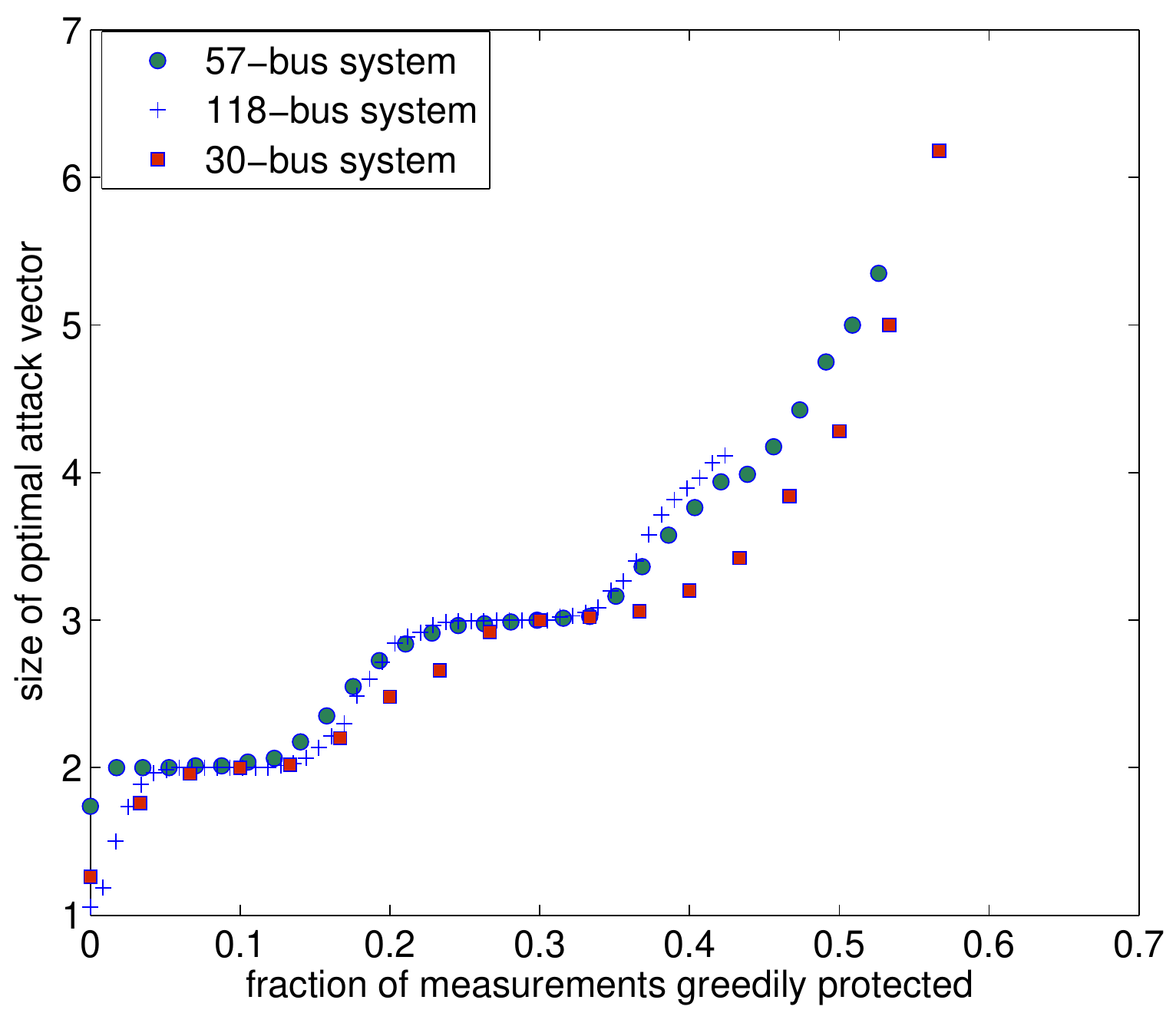}
\caption{Greedy protection of additional measurements in IEEE test systems with flow measurements on all lines, voltage measurements on $60\%$ of the buses and $1/6$ of measurements initially protected}
\label{measprotgreedy}
\end{figure}
\begin{figure}
\centering
\includegraphics[width=0.4\textwidth]{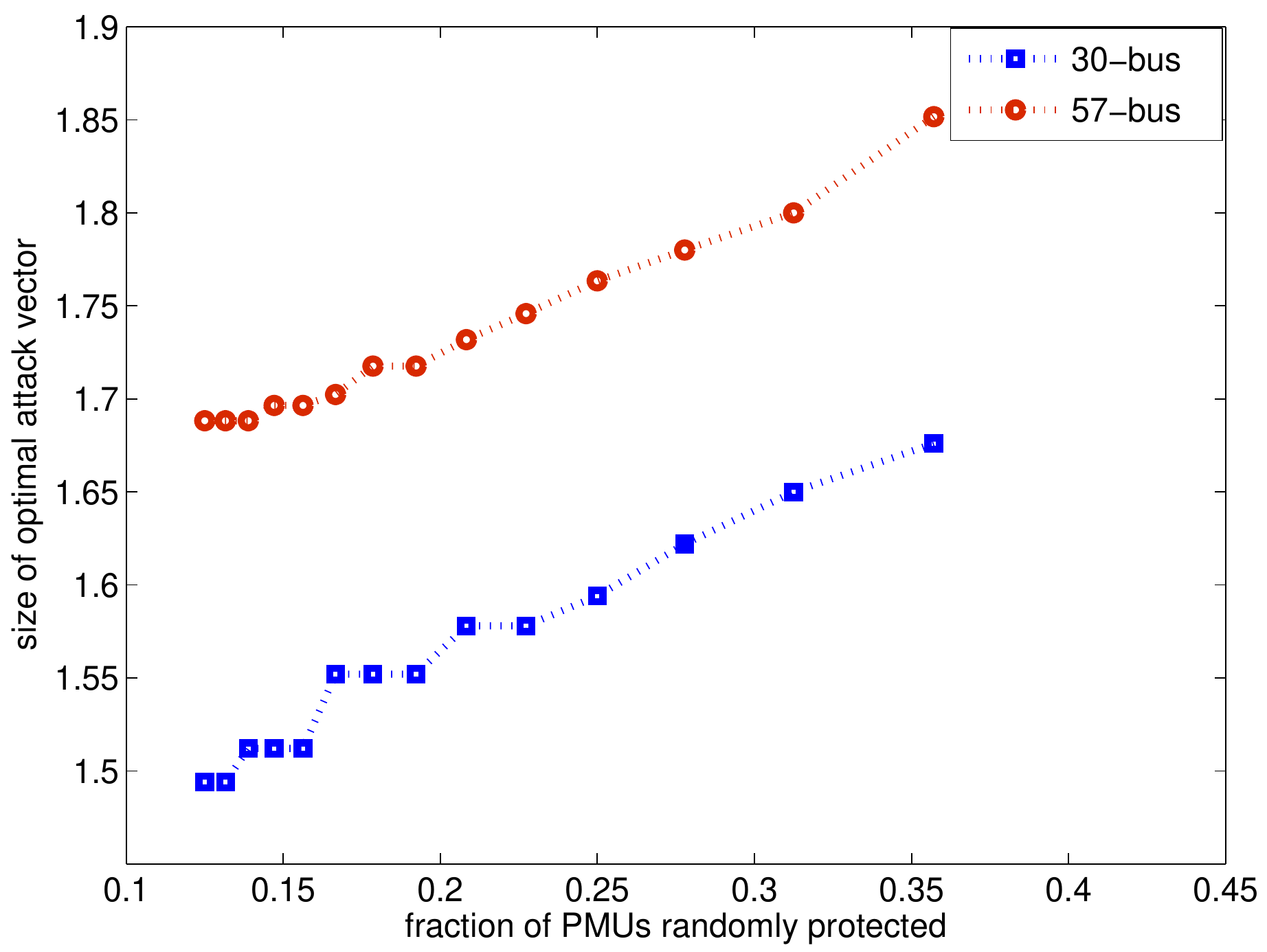}
\caption{Optimal hidden attack on IEEE test systems with flow measurements on all lines, voltage measurements on $60\%$ of the buses and PMUs placed randomly on fraction of buses}
\label{attackpmu}
\end{figure}
\begin{figure}
\centering
\includegraphics[width=0.4\textwidth]{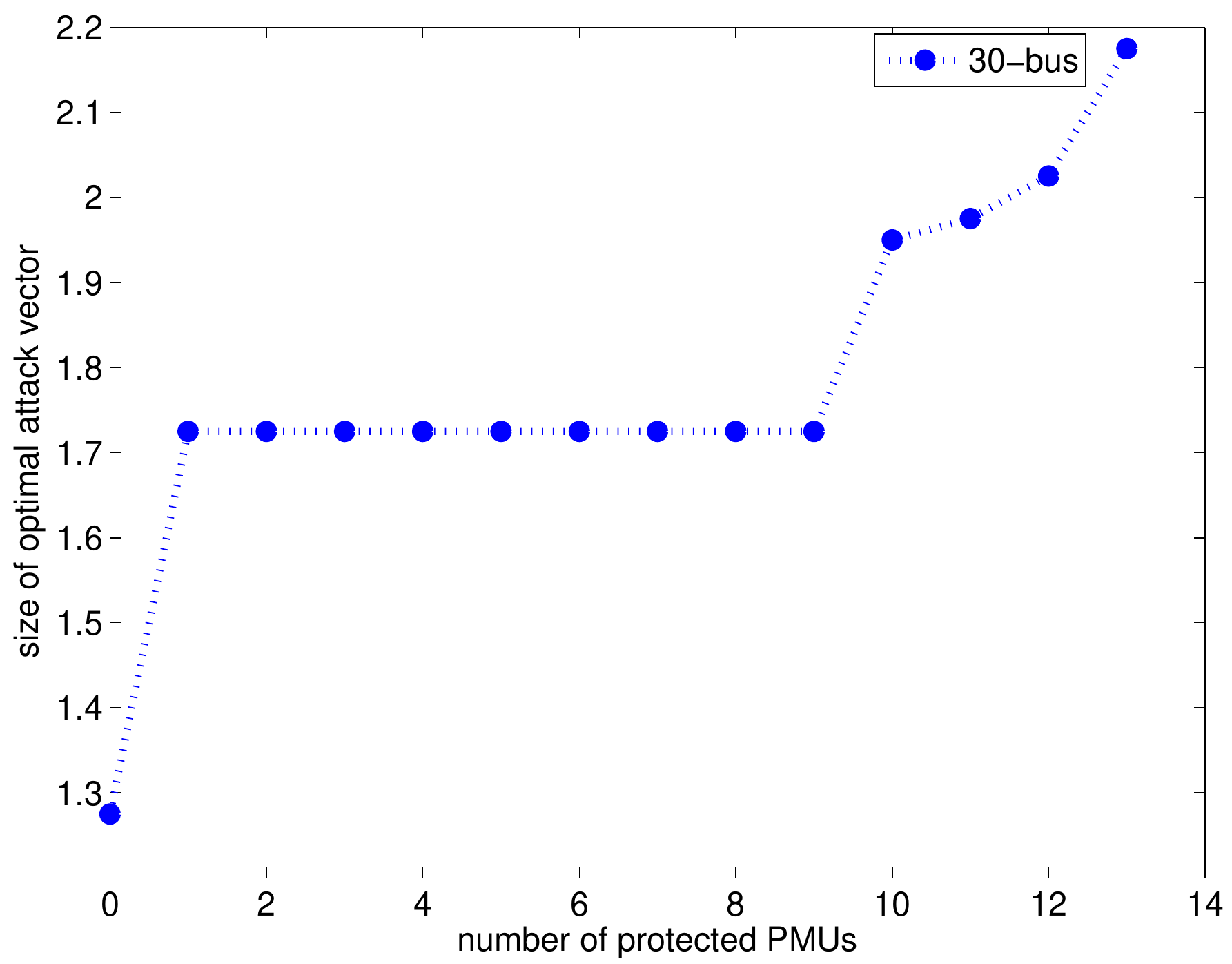}
\caption{Greedy placement of additional PMUs in IEEE 30-bus system with flow measurements on all lines, voltage measurements on $60\%$ of the buses and $1/10$ of measurements initially protected}
\label{PMUprot30}
\end{figure}

\section{Conclusion}
\label{sec:conclusion}
In this paper, we study an adversarial problem of causing errors in estimation of state variables in a power grid through injection of suitably hidden measurement errors. We formulate this problem in terms of controlling and manipulating a minimum set of measurements in order to affect a successful hidden attack as an $l_0$ optimization problem. We introduce a novel graph-theoretic approach to designing the optimal attack vector using min-cuts. The proposed algorithm has polynomial time complexity and is shown to result in an optimal output given a configuration of the power grid. We show that our algorithm gives the optimal output even when a fraction of the measurements have existing protection and performs much better than a $l_1$ relaxation of the problem.
From the system operator's perspective, we develop an algorithm to identify measurements in the system that provide additional protection, aimed at preventing and/or reducing the efficacy of hidden attacks by an adversary. Although sub-optimal, the low complexity algorithm can be used to protect measurements to increase the set of measurements that the adversary must control in order to cause a successful hidden attack on the system. Further, we extend the discussion on hidden attacks in the grid to systems with PMUs and discuss design of optimal attack vector for a system with PMUs and placement of additional secure PMUs in the system to prevent such attacks. The advantage of using low complexity algorithms to provide security against hidden attacks is immense for large power grids with several thousand buses and lines. This work can be extended to include other hidden attacks where the adversary is not limited by number of attacked meters but other resources. Another extension includes determining the minimum set of key measurements for protection by the system operator given the knowledge of the adversary's maximum capacity to attack the power grid. This is the focus of our current work.

\end{document}